\documentclass[10pt, conference]{IEEEtran}
\IEEEoverridecommandlockouts
\usepackage{amsfonts,amsmath,amssymb}
\newtheorem{thm}{Theorem}%[section]
\newtheorem{cor}{Corollary}
\newtheorem{lem}{Lemma}
\newtheorem{prop}{Proposition}

\newtheorem{defn}{Definition}%[section]
\newtheorem{example}{Example}%[subsection]
\newtheorem{remark}{Remark}
\newcommand{\R}{\mathbb{R}}
\newcommand{\C}{\mathbb{C}}
\newcommand{\F}{\mathbb{F}}
\newcommand{\N}{\mathbb{N}}
\newcommand{\eps}{\epsilon}
\author{ \authorblockN{Ehsan Ardestanizadeh\thanks{ E.~Ardestanizadeh
      was supported in part by NSF Career Award CNS-0533035, and by the
      Center for Networked Systems grant CNS0512 at UCSD. Work done while
      he was visiting Laboratory of Algorithms (ALGO) at EPFL. 
      M.~Cheraghchi was supported by Swiss NSF grant
      200020-115983/1. }}
  \authorblockA{%Department of Electrical Engineering\\
    ECE, UCSD\\
    La Jolla, CA, 92093-0407, USA \\
    eardesta@ucsd.edu} \and \authorblockN{Mahdi Cheraghchi}
  \authorblockA{%Lab. of Math. Algorithms \\
    ALGO, EPFL\\
    1015 Lausanne, Switzerland\\
    mahdi.cheraghchi@epfl.ch } \and \authorblockN{Amin Shokrollahi}
  \authorblockA{%Lab. of Math. Algorithms \\
    ALGO/LMA, EPFL\\
    1015 Lausanne, Switzerland\\
    amin.shokrollahi@epfl.ch } }
\date{}

\title{Bit Precision Analysis for Compressed Sensing}

\begin{document}
\maketitle
\begin{abstract}
  This paper studies the stability of some reconstruction algorithms
  for compressed sensing in terms of the \emph{bit
    precision}. Considering the fact that practical digital systems
  deal with discretized signals, we motivate the importance of the
  total number of accurate bits needed from the measurement outcomes
  in addition to the number of measurements.  It is shown that if one
  uses a $2k \times n $ Vandermonde matrix with roots on the unit
  circle as the measurement matrix, $O(\ell + k \log \frac{n}{k})$
  bits of precision per measurement are sufficient to reconstruct a
  $k$-sparse signal $x \in \R^n$ with \emph{dynamic range} (i.e., the
  absolute ratio between the largest and the smallest nonzero
  coefficients) at most $2^\ell$ within $\ell$ bits of precision,
  hence identifying its correct support. Finally, we obtain an upper
  bound on the total number of required bits when the measurement
  matrix satisfies a restricted isometry property, which is in
  particular the case for random Fourier and Gaussian matrices.  For
  very sparse signals, the upper bound on the number of required bits
  for Vandermonde matrices is shown to be better than this general
  upper bound.
\end{abstract}
% This paper studies the stability analysis of syndrome decoding
% algorithm for Reed-Solomon codes when it is exploited for signal
% reconstruction in compressed sensing. The goal is to sense a
% $k$-sparse signal $x \in \R^n$ with dynamic range of $2^\ell$
% through quantized observations. It is shown that if one uses a $2k
% \times 2k $ Vandermonde matrix with roots on the unit circle as the
% measurement matrix, $2\ell+O(k \log \frac{n}{k} )$ bit precision per
% measurement is enough to reconstruct the signal with correct
% support. Considering the fact that dealing with real numbers need
% infinitely many bits and any practical system works with finite
% precision, it is shown that the total number of precision bits in
% the measurements is a more practical measure for optimality of an
% algorithm rather than the total number of measurements. A comparison
% between Vandermonde matrices combined with syndrome decoding and
% random Fourier measurements combined with $\ell_1$ minimization
% shows that in the regime where $k=O(\log n)$ less total number of
% measurement bits is sufficient for the former algorithm to ensure
% correct support recovery.

\section{Introduction}

Compressed sensing is an emerging field which deals with new sampling
techniques for sparse signals. The goal is to exploit the sparsity of
the signal and try to reconstruct the signal using a number of linear
measurements far below the signal's dimension. Formally, let $x \in
\R^n$ be a $k$-sparse vector ($k \ll n$) and $y = Ax,$
% \begin{equation*}
%   y=Ax,
% \end{equation*}
where $A$ is an $m \times n$ measurement matrix with possibly complex
entries, and $y \in \C^m$ is the observed vector. The main problem is
to design a measurement matrix with $m \ll n$ for which there exists
an efficient reconstruction algorithm that is able to reconstruct any
$k$-sparse signal from the measurement $y$. It is easy to check that
at least $2k$ measurements are required if \emph{all} $k$-sparse
signals have to be distinguishable. However, it can be shown that
$k+1$ linear measurements with random coefficients are enough to
reconstruct \emph{almost} all signals almost surely
\cite{l_0,l_00}. To reconstruct the signal we need to solve the
following $\ell_0$ minimization problem

% \begin{lem}\label{rank}
%   If matrix $A_{m\times n}$ with $m < n$ is able to distinguish all
%   $k$-sparse signals then $m \geq 2k$.
% \end{lem}
% \begin{proof}
%   Let $x_1$ and $x_2$ be two $k$-sparse signals. Then $x_1 - x_2$
%   would be a $2k$-sparse signal. If we do not want to confuse the
%   two signals then we need to have $A(x_1-x_2) \neq 0$. Therefore
%   any $2k$ columns of $A$ have to be linearly independent and
%   $rank(A) \geq 2k$. On the other hand we know that rank(A) $\leq
%   \min(m,n)=m$, so $m \geq 2k$
% \end{proof}
% Lemma \ref{rank} states that for $A$ to be good matrix for sensing
% $k-$sparse vectors we need at least $2k$ measurements.
\begin{equation*}
  (P_0)    \qquad   \min \|\hat{x}\|_{0},   \qquad  \qquad A \hat{x}=y,    %\label{p0}
\end{equation*}
which is known to be $\mathsf{NP}$-hard for a generic matrix $A$.

A key observation by Candes et al.\ \cite{Candes,Candes1,Candes2} and
Donoho \cite{Donoho} shows that if matrix $A$ obeys a so-called
\emph{restricted isometric property} (RIP)\footnote{This property was
  originally called \emph{uniform uncertainty principle} (UUP) by
  Candes and Tao .}, which essentially requires that any set of up to
$k$ columns of $A$ behaves close to an orthonormal system, then the
signal can be exactly reconstructed using the following $\ell_1$
minimization program
\begin{equation*}
  (P_1)  \qquad     \min \|\hat{x}\|_1,   \qquad  \qquad   A \hat{x}=y.  %\label {P_1}
\end{equation*}

This is an easier problem compared to $\ell_0$ minimization and in
particular can be solved in polynomial time %$O(n^3)$
using linear programming (LP) techniques. There are families of random
matrices which satisfy the RIP with high probability if $m$, the
number of rows, is large enough. Two examples of such families are
given by random Gaussian measurements and random Fourier
measurements. If $A$ is a random matrix with i.i.d.\ Gaussian entries
and $m=O(k\log \frac{n}{k})$, or if $A$ is constructed from $m=O(k\log
n)$ random rows of the $n \times n$ discrete Fourier transform matrix,
then the matrix can be shown to satisfy the RIP with high probability
\cite{Candes}. Even though the results for random matrices hold with
high probability, there is no known efficient way to verify if a
random matrix satisfies the RIP.  This motivates the problem of
finding an \emph{explicit} construction of a measurement matrix $A$
with small number of measurements for which we can solve $(P_0)$
efficiently. It is shown that explicit matrices can be constructed
based on group testing techniques \cite{Group} as well as expander
graphs and randomness extractors \cite{Hassibi1, Hassibi2, Jafarpour,
  Indyk, Indyk2}.

A closer inspection reveals that $(P_0)$ is an analog of the so-called
\emph{syndrome decoding} problem over real or complex numbers, and
hence it is natural to expect that known techniques from coding theory
might be applicable in compressed sensing. In particular, Akcakaya and
Tarokh \cite{Tarokh} show that several results known for the
Reed-Solomon codes over finite fields can be extended to the field of
complex numbers. Therefore, similar coding and decoding algorithms can
be used for sensing sparse vectors over the real or complex
field. Specifically, they show that it is possible to reconstruct any
$k$-sparse vector from only $2k$ measurements, which is the minimum
number of measurements one can hope for, using $O(n^2)$ arithmetic
operations.

% While the above results are stated assuming that the signal $x$ is
% sparse in the canonical basis, it is easy to see that they can be
% easily generalized to sparse vectors in arbitrary bases. Namely,
% suppose that the vector $x$ is known to be sparse in some basis
% $\{\psi_i\}$. Let ${\Psi}$ denote the $n \times n$ matrix with
% $\psi_i$ as its $i$th row and $A$ be a measurement matrix for sparse
% vectors in the canonical basis. Then $\tilde{x}:={\Psi}x$ would be
% sparse (in the canonical basis) and we can use $A$ to identify
% $\tilde{x}$. Thus, the measurement matrix $A\Psi$ can be used on $x$
% to identify $\tilde{x}$, from which we can obtain
% $x=\Psi^{-1}(\tilde{x})$.

All the above results hold under the assumption that measurements and
arithmetic over real numbers are carried out precisely. However, in
digital systems we generally cannot deal with real numbers simply
because we would need infinitely many bits to represent a real
number. So it is inevitable to resort to truncated representations of
real vectors. Thus a natural question to ask is how precise the
measurement outcomes need to be so as to be able to reconstruct the
original data within a \emph{target precision}. In
Section~\ref{bitprecision}, we will show how it becomes important to
not only take the total number of measurements into account, but also
the precision required from individual measurements. Together, these
two quantities give a suitable measure of the amount of
\emph{information} (in bits) that needs to be extracted from the
measurements in order to approximate the sparse signal. We will use a
simple example to justify the point that if the precision of the
measurements is allowed to be sufficiently high, even \emph{one}
measurement is sufficient to reconstruct discrete signals.

The main result of this paper is a bit precision analysis of the
syndrome decoding algorithm for Reed-Solomon codes when applied in the
context of compressed sensing as a reconstruction algorithm for
Vandermonde measurement matrices \cite{Tarokh}. The analysis is based
on the assumption that the sparse signal is to be reconstructed within
a certain chosen precision in the fixed-point model and the additional
requirement that the support of the reconstructed vector is the same
as that of the original signal. In particular, we show that if the
dynamic range of $x$ is at most $2^\ell$ then having each measurement
available within $O(\ell+k\log \frac{n}{k})$ bits of precision is
sufficient to identify $x$ within $\ell$ bits, which is the minimum
precision needed to ensure that the smallest nonzero entry of $x$ is
not confused with zero. Since we have a total of $2k$ measurements,
the total number of bits required from the measurement outcomes is
upper bounded by $O(\ell k+k^2\log \frac{n}{k})$.
% A comparison shows that in the regime where $k=O(\log n)$ this is
% less than the upper bound on the total number of required bits for
% random Fourier measurements combined with $\ell_1$ minimization .
% Note that we need at least $\ell$ bit precision for nonzero entries
% so that we do not miss any of them and recover $x$ with the right
% support, and also we need $O(k\log \frac{n}{k})$ bits to represent
% the places of the nonzero elements, so $kl+O(k\log \frac{n}{k})$ .
% , which is close to the theoretical lower bound $O(kl+k\log
% \frac{n}{k})$ if $k \ll n$. The

The rest of the paper is organized as follows. First, in
Section~\ref{bitprecision} we motivate the total bit precision as a
practical measure for assessing the quality of compressed sensing
algorithms. Then, in Section~\ref{definition} we give a more rigorous
definition of the problem that we consider and state our stability
theorem for Vandermonde measurements and the syndrome decoding
algorithm. Section~\ref{stability} gives the sketch of the proof for
the stability theorem and in Section~\ref{comparison} we will upper
bound the total number of bits required from the measurements obtained
from matrices satisfying certain restricted isometry properties.
% a comparison between syndrome decoding and $\ell_1$
% minimization. %Finally, Section \ref{conclusion} concludes the paper.

\section{The Importance of Bit Precision}\label{bitprecision}

The main purpose of this work is to show that while the number of
measurements is an important criterion for assessing the quality of a
compressed sensing scheme, it is by itself insufficient without
considering the precision needed for the reconstruction algorithm to
work properly. Indeed, a more favorable approach than simply bounding
the number of measurements would be to quantify the total amount of
\emph{information} (in bits) that needs to be extracted from the
measurement outcomes so as to enable a reliable reconstruction of the
original signal within a pre-specified precision.  Intuitively, a
single real number can pack an infinite amount of information and for
virtually all real world applications, either the signal to be
measured is \emph{a priori} known to be discrete (for example, the
output of a sensor measuring the temperature over a long period of
time), or is only needed within a certain pre-specified number of
accurate bits. For all such cases, a single measurement is in
principle capable to carry all the needed information. The following
example illustrates this point.

\begin{example}
  Suppose that $A$ is an $m \times n$ binary matrix that allows
  recovery of $k$-sparse vectors over $\F_2$.  Such a matrix can be
  obtained from a parity check matrix of a binary code with minimum
  distance at least $k+1$.  We now ``compress'' the matrix $A$ into a
  vector $a = (a_1, \ldots, a_n) \in \R^n$ such that any $k$-sparse $x
  \in \{0,1\}^n$ can be exactly reconstructed from $a \cdot x \in
  \R$. We define
  \[
  a_i := \sum_{j=1}^{m} A(j, i) \cdot 2^{j(\lceil \log (k+1) \rceil)},
  \]
  where $A(j, i)$ denotes the the entry of $A$ at the $j$th row and
  the $i$th column.  This vector simply encodes all the rows of the
  matrix $A$ by shifting each row by a sufficient amount to prevent
  any confusion, and a moment's thought reveals that indeed $x$ can be
  uniquely reconstructed\footnote{ Using a similar construction, it is
    also easy to see that, allowing infinite precision in the
    measurements, it is information theoretically possible to uniquely
    identify any (not necessarily sparse) discrete vector $x \in \N^n$
    using only one real linear measurement.  } from $a \cdot x$.
  However, by a simple counting, the number of rows of $A$ has to be
  at least $\log \binom{n}{k} = \Omega(k \log(n/k))$, and we need at
  least $m \times \lceil \log (k+1) \rceil$ bits from $a \cdot x$ to
  be able to reconstruct $x$.  Hence, although the number of
  measurements is extremely low, the total number of bits that we need
  to extract from the measurement has to be at least $\Omega(k \log k
  \log(n/k))$.
\end{example}

In general, a counting argument shows that
% \emph{simple}
matrices with entries from a small domain cannot be used to bring down
the number of measurements below a certain level. This is captured in
the proposition below:
% \footnote{ Here we only consider matrices with integer entries, but
%   obviously the argument also holds for matrices with rational
%   entries as they can be transformed to integer matrices by an
%   appropriate scalar multiplication. }:

\begin{prop}
  Let $A$ be an $m \times n$ matrix whose entries are integers in
  range $[-2^\ell, 2^\ell]$.  Assume that $A$ can be used for
  reconstruction of $k$-sparse signals in $\R^n$. Then $m =
  \Omega\left(\frac{k \log (n/k)}{\ell + \log k}\right)$.
\end{prop}

\begin{proof}
  The matrix $A$ must be in particular able to distinguish binary
  $k$-sparse vectors.  The number of such vectors is
  $\binom{n}{k}$. Let $x$ be a $k$-sparse binary vector and $y :=
  Ax$. Each entry of $y$ must be an integer in range $[-k 2^\ell, k
  2^\ell]$, and the number of vectors in $\R^n$ satisfying this
  property is $(k 2^{\ell+1} + 1)^m$, and this number must be lower
  bounded by the number of $k$-sparse binary vectors. This gives the
  desired bound.
\end{proof}

The above result explains why the entries of our single-measurement
matrix had long binary representations. However, as shown in
\cite{Tarokh}, one can ``break'' this lower bound using Vandermonde
matrices and achieve a total of $2k$ measurements. This special
property of Vandermonde matrices is due to the fact that the entries
of the matrix cannot be represented by bounded precision numbers and
the amount of required precision must necessarily grow to infinity as
$n$ gets large. Hence, Vandermonde matrices use large precision in an
essential way and it becomes a crucial task to quantitatively analyze
the amount of precision that Vandermonde measurements need for making
reliable reconstruction of sparse signals possible.

\section{Problem Definition and Main result}\label{definition}
We consider the problem of recovering a $k$-sparse signal $x \in \R^n$
from $m \ll n$ discretized observations. Let $A$ be an $m \times n$
matrix with possibly complex entries and $\hat{y} := y + e \in \C^m$,
be our observation vector,
% as follows
% \begin{equation*}
%   y=Ax+e,
% \end{equation*}
where $y := Ax$ and $e \in \C^m$ is the truncation noise in the
observation. Throughout the paper we consider the fixed-point binary
representation of real numbers and define precision as follows.
\begin{defn}\label{deltadef}
  We say that a vector $\hat{z}\in \C^n$ is an approximation of $z\in
  \C^n$ within $\ell$ bits of precision (or $\ell$ accurate bits) if
  $\|\hat{z}-z\|_{\infty}/\|z\|_{\infty} < 2^{-\ell}.$
  % \[\frac{\|\hat{z}-z\|_{\infty}}{\|z\|_{\infty}} < 2^{-\ell}.\]
\end{defn}
\begin{defn}\label{def:dynrange}
  The dynamic range of a nonzero vector $x$ is defined as the ratio
  $|x_{\max}/x_{\min}|,$ where $x_{\max}$ and $x_{\min}$ are the
  largest and the smallest nonzero entries of $x$ in absolute value,
  respectively.
\end{defn}
The problem is to find the sufficient precision for $y$ such that we
can ensure that a signal $x$ with known dynamic range can be recovered
with the correct support. This of course depends on the matrix $A$ and
the reconstruction algorithm. We pick as $A$ a Vandermonde matrix with
roots on the unit circle; namely,
\begin{equation*}
  A := \left( \begin{array}{cccc}
      1       &    1& \dots & 1 \\
      a_{1} & a_{2} & \ldots & a_{n} \\
      a_{1}^2 & a_{2}^2& \ldots & a_{n}^2 \\
      \vdots & \vdots & \ddots & \vdots \\
      a_{1}^{m-1} & a_{2}^{m-1} & \ldots & a_{n}^{m-1}
    \end{array} \right),
\end{equation*}
where $a_j:=\exp{(j\frac{2\pi\sqrt{-1}}{n})}$. This choice of $A$ is
motivated by Reed-Solomon codes over finite fields. By the properties
of Reed-Solomon codes if the ``error pattern'' (i.e., the vector $x$)
is $k$-sparse and we pick $m=2k$ then it can be uniquely identified
from the measurement outcomes.  As noted in \cite{Tarokh} this
property holds over the complex field as well and therefore if we use
a Vandermonde matrix with distinct roots as the sensing matrix, we can
exploit an analog of Reed-Solomon decoding algorithm over the complex
field to reconstruct the signal. The Reed-Solomon decoding algorithm
we use is the so-called \emph{syndrome decoding} algorithm, where (for
the case $e = 0$) the measurements $y_0, y_1, \ldots, y_{m-1}$ are
considered as syndromes from which we wish to find the corresponding
error pattern (i.e., the vector $x$).  The decoding algorithm is as
follows. First, we solve the following Toeplitz linear system
\begin{equation}\label{toeplitz}
  \left( \begin{array}{cccc}
      y_0 &  y_1& \dots & y_{k}  \\
      y_1   & y_2 & \ldots & y_{k+1} \\
      \vdots & \vdots & \ddots & \vdots \\
      y_{k-1} & y_{k} & \ldots & y_{2k-1}
    \end{array} \right)  \left( \begin{array}{cccc}
      h_{k}  \\
      h_{k-1}  \\
      \vdots  \\
      h_{0}
    \end{array} \right) = \left( \begin{array}{cccc}
      0  \\
      0  \\
      \vdots  \\
      0
    \end{array} \right)
\end{equation}
for a nonzero solution, and let $h(x):= h_0 + h_1 x+ \cdots + h_k
x^k$. From the theory of Reed-Solomon codes (cf.\ \cite{Coding}) we
know that $h(x)$ is a multiple of the \emph{error locator polynomial}
$L(x):=\prod_{e\in E} (1-xa_{e})$, where $E \subseteq \{1, 2, \ldots,
n\}$ is a set of size at most $k$ containing the error positions
(i.e., the support of $x$). Therefore the set of the zeros of $h(x)$
determines a superset of the error positions (and the exact set if $h$
is a nonzero solution with the smallest degree). Having found a
superset of error locations with size $k$, we can solve a $k \times k$
system of linear equations to find the actual error values.
% So this method achieves the lower bound on the number of ``real''
% measurements given in Lemma~\ref{rank}.

In the presence of truncation noise, it is natural to consider the same
reconstruction method using the truncated syndrome vector $\hat{y}$,
and ask how stable the method is in this situation. The following theorem
quantifies the amount of precision of the measurements needed to
ensure correct recovery of the support:
\begin{thm}\label{theor}
  If we use a $2k \times n$ Vandermonde matrix with roots on the unit
  circle and observe the syndromes within $O(\ell+k\log \frac{n}{k})$
  accurate bits, then we can stably reconstruct any $k$-sparse signal
  with dynamic range at most $2^\ell$. The reconstructed signal has
  the same support as the original signal and approximates its nonzero
  elements within $\ell$ bits of precision.
\end{thm}

The proof of this theorem is sketched in Section~\ref{stability}.

\begin{remark} For the fixed point model that we are considering in
  this work, if the signal has dynamic range $2^\ell$, then we
  obviously need at least $\ell$ bits of precision in the
  reconstructed signal to make sure that all the nonzero entries are
  being recovered with nonzero magnitudes.
\end{remark}

The theorem states that if we want to recover the correct support of a
signal with dynamic range $2^\ell$, $O(\ell+k\log \frac{n}{k})$ bits
of precision per measurement is sufficient. Since we have $2k$
measurements, the total number of measurement bits add up to $O(k\ell
+ k^2\log \frac{n}{k})$.  In particular, if we consider binary signals
for which $\ell=0$, then we will require $O\big(k^2 \log
\frac{n}{k}\big)$ bits in total.  Nevertheless, this bound is $k$
times larger than $O\big(k \log \frac{n}{k}\big)$, the information
theoretical lower bound on the number of required bits to identify
$k$-sparse binary signals.
% We will elaborate on the bounds for the total number of bits in
% Section \ref{comparison}.

% In Section~\ref{comparison} we will show that if we want to ensure
% correct recovery of the support and $k=O(\log n)$, this is better
% than the upper bound on the total number of measurement bits
% required for $\ell_1$ minimization when random Fourier coefficients
% are used for the measurements.  \marginpar{Amin: Here we need to
%   discuss how much larger than the lower bound we are.}

 \section{Sketch of the Proof of Theorem~\ref{theor}}\label{stability}
 In this section we outline the proof of Theorem~\ref{theor} and omit
 certain details due to space restrictions. In particular, we focus on the
 case where the support size of the sparse signal is exactly $k$ or
 known to the decoder\footnote{If the decoder knows the actual support
   size $t$, where $t < k$, it can discard all but $2t$ of the
   syndromes and reduce the problem to the case where the support size
   is exactly half the number of measurements.  Otherwise, it can try
   various possibilities for $t$ and find a list of up to $k$ possible
   reconstructions that includes a correct approximation of $x$. It is
   not hard to see that only one of these reconstructions can
   reproduce the given syndromes within a sufficient precision, and thus,
   the decoder can always uniquely reconstruct $x$.  }. The proof is
 done in three steps. First, we find sufficient precision for the
 $h_i$ so as to be able to detect the positions of nonzero elements,
 namely the roots of the error locator polynomial. Next, an upper
 bound is derived on the required precision for the $y_i$ so as to
 guarantee that the $h_i$ can be solved, with desired precision, from
 the set of linear equations given in~\eqref{toeplitz}. This upper
 bound depends on the condition number of a $k \times k$ matrix with
 entries given by the $y_i$, for which an upper bound is derived in
 the last step.

 {\it Step~1: } In the presence of noise, $h_0, h_1, \ldots, h_k \in
 \C$ are noisy and $h(x)$, assuming that $h_0 = 1$, is not necessarily
 equal to the error locator polynomial $L(x)$. Thus, we need to make
 sure that the error in the $h_i$ is small enough to allow us to
 reliably find the roots of $L(x)$. From the choice of $L(x)=\prod_{e
   \in E} (1-xa_e)$, the minimum nonzero magnitude of $L(x)$ can be
 bounded as
 \[|L(x)| \geq k!\left(\frac{2\pi}{n}\right)^k, \] since each $a_e$
 and the evaluation point $x$ is an $n$th root of unity and the
 quantity $| 1-xa_e |$ is the length of a chord\footnote{Here for the
   sake of clarity we are neglecting the lower order term $O(1/n^2)$
   in the approximation on the length of the chord, but it should be
   clear that this will not affect the analysis.} on the unit circle
 whose corresponding angle is a distinct multiple of $2\pi/n$ (as the
 $a_e$ are distinct).

 Also, the magnitude of the error in evaluation of $h(x)$ can be upper
 bounded by $k^2 2^{-\ell_h}$ if the $h_i$ are available within
 $\ell_h$ bits of precision.  This is because $L(x)$ has $k$ monomials
 with coefficients of magnitude at most $k$.

 Hence, to find the roots of $L(x)$ correctly, it suffices to have
 \[
 k^2 2^{-\ell_h} < k!\left(\frac{2\pi}{n}\right)^{k},
 \]
 which can be satisfied by having the coefficient vector of $h(x)$
 within $\ell_h = O(k\log{\frac{n}{k}})$ bits of precision.

 {\it Step~2: } Now that we have a bound on the precision that we need
 for the $h_i$, we have to calculate the
 precision we need for $y$. We will use the following theorem.
\begin{thm}\cite[Theorem 7.2]{Numerical} %\label{num}
  Let $Ax=b$, where $A$ is a square invertible matrix %with real entries
  and $x$ and $b$ are vectors, and $(A+\Delta
  A)y=b+\Delta b$, where $\|\Delta A\| \leq \epsilon \|E\| $ and
  $\|\Delta b \| \leq \epsilon \|f\|$. The matrix $E$ and the vector
  $f$ are arbitrary and $\|\cdot\|$ is any absolute norm. Furthermore,
  assume that $\epsilon\|A^{-1}\| \ \|E\| < 1$. Then
  \begin{equation*}
    \frac{\|x-y \|}{\|x \|} \leq \frac{\epsilon}{1-\epsilon\|A^{-1}\| \ \|E\|}\left(\frac{\|A^{-1}\| \ \|f\|}{\|x\|}+\|A^{-1}\| \ \|E\|\right).
  \end{equation*} %\qed
\end{thm}

By picking $E:=A$, $f:=b$ and $\|\cdot\|_{\infty}$ as the norm function,
we obtain the following corollary:

\begin{cor} \label{num}
  For a square matrix $A$ let $Ax=b$ and $(A+\Delta A)y=b+\Delta b$,
  where $\|\Delta A\|_{\infty} \leq \epsilon \|A\|_{\infty} $ and
  $\|\Delta b \|_{\infty} \leq \epsilon \|b\|_{\infty}$, and assume
  that $\epsilon \kappa_{\infty}(A) \leq \frac{1}{2}$. Then
  \begin{equation*}
    \frac{\|x-y \|_{\infty}}{\|x \|_{\infty}} \leq 4\epsilon \kappa_{\infty}(A), % + o(\epsilon^2)
  \end{equation*}
  where $\kappa_{\infty}(A)$ denotes the condition number of the matrix
  $A$ with respect to the $\infty$-norm.
\end{cor}

%  \begin{thm}\cite[Theorem 7.2]{Numerical}\label{num}
%    For a square real matrix $A$ let $Ax=b$ and $(A+\Delta A)y=b+\Delta
%    b$, where $\|\Delta A\|_{\infty} \leq \epsilon \|A\|_{\infty} $ and
%    $\|\Delta b \|_{\infty} \leq \epsilon \|b\|_{\infty}$, and assume
%    that $\epsilon \kappa_{\infty}(A) \leq \frac{1}{2}$. Then $\|x-y
%    \|_{\infty}/\|x \|_{\infty} \leq 4\epsilon \kappa_{\infty}(A)$,
%    % \begin{equation*}
%    %   \frac{\|x-y \|_{\infty}}{\|x \|_{\infty}} \leq 4\epsilon
%    %   \kappa_{\infty}(A), % + o(\epsilon^2)
%    % \end{equation*}
%    where $\kappa_{\infty}(A)$ denotes the condition number of matrix
%    $A$ with respect to the $\infty$-norm.
%  \end{thm}
 The corollary states that if we wish to obtain the solution $\hat{x}$
 of a linear system $Ax=b$ up to $r$ accurate bits, i.e.,
 $\|\hat{x}-x \|_{\infty} / \|x \|_{\infty} < 2^{-r}$, it
 suffices to have $\|\Delta A\|_{\infty} \leq \epsilon \|A\|_{\infty}
 $ and $\|\Delta b \|_{\infty} \leq \epsilon \|b\|_{\infty}$ with
 $\epsilon=O(2^{-(r+\log\kappa_{\infty}(A))})$.
 % In our problem the entries of matrix $A$ and vector $b$ are chosen
 % from the syndromes $y_0, \ldots, y_{2k-1}$. Since $$we require a
 % relative error of $O(2^{-(r+\log\kappa_{\infty}(A))})$ in the
 % syndromes, $r+\log\kappa_{\infty}(A)+O(1)$ correct significant bits
 % in the fixed-point representation of $y_{\min}$, the smallest
 % nonzero element of $y$, would be sufficient to achieve the desired
 % precision.

From Step~1, we know that we need $O(k\log{\frac{n}{k}})$ bits of
precision for the $h_i$. Now suppose that the decoder receives
a perturbed version of the syndrome vector and finds a nonzero
solution for the (perturbed) system of linear equations \eqref{toeplitz},
namely, $\hat{h} := (\hat{h}_0, \ldots, \hat{h}_k)$. We show that there is a solution
$h := (h_0, \ldots, h_k)$ for the original system \eqref{toeplitz}
that is sufficiently close to the perturbed solution, i.e.,
$\| \hat{h} - h \|_\infty / \| h \|_\infty < 2^{-\Omega(k\log{\frac{n}{k}})}$,
as required by Step~1.

\newcommand{\ymax}{y_{\mathrm{max}}}
Denote by $M$ the coefficient matrix of \eqref{toeplitz} (prior to
the perturbation of syndromes),
and by $M_i$ the $k \times k$ minor of $M$ obtained by removing
the column corresponding to $h_i$. Moreover, define $\ymax$ as the
largest syndrome in absolute value so that $|\ymax| = \|y\|_\infty$, and
note that each $M_i$ contains all the syndromes but one.
It is always possible to set $h_i = \hat{h}_i$,
for some choice of $i$ such that $M_i$ contains 
an entry with magnitude $|\ymax|$, and we can rewrite \eqref{toeplitz}
as
\begin{equation} \label{lnreqn}
\begin{split}
M_i (h_k, \ldots, h_{i+1}, h_{i-1}, \ldots, h_0)^\top = \hspace{2cm}
\\ -\hat{h}_i(y_{k-i}, y_{k-i+1}, \ldots, y_{2k-i-1})^\top. 
\end{split}
\end{equation}

Here we mention a technicality that the precision of the $y_i$ is
bounded relative to the largest coefficient $\ymax$, which is not
necessarily present on the right hand side of the above
system. However, we can add an additional ``dummy'' equation $\ymax
h_i = \hat{h}_i \ymax$ to the system and ensure that the
requirements of Corollary~\ref{num} on the error bounds are fulfilled.
It is easy to see that the new system will have condition number at most
$\max\{\kappa_\infty(M_i), k\}.$ We can now apply Corollary~\ref{num}
on the system given by \eqref{lnreqn} and find a sufficient precision for the $y_i$.
In particular, we conclude that 
 \begin{equation}
   O(k\log(n/k))+ \log\kappa_{\infty}(M_i) \label{bits}
 \end{equation}
 bits of precision for the $y_i$ would be sufficient for finding the $h_i$
 within the precision required by Step~1, and thus, the
 correct support of the sparse vector.

%  From Step~1 we know that we need $O(k\log{\frac{n}{k}})$ bits of
%  precision for the $h_i$. Since we are interested in the roots of
%  $h(x)$, we can fix one of the $h_i$ (e.g., $h_0=1$) and rewrite
%  \eqref{toeplitz} as $Mh = b$, where $M$ is a $k\times k$ minor of the
%  the coefficient matrix of \eqref{toeplitz}, $h$ is the vector of
%  unknown coefficients of $h(x)$, and $b$ is a vector of
%  syndromes. Then we apply Theorem~\ref{num} (ignoring some
%  technicalities) and conclude that
%  \begin{equation}
%    O(k\log(n/k))+ \log\kappa_{\infty}(M) \label{bits}
%  \end{equation}
%  bits of precision for $y_i$ would be sufficient for finding the
%  correct support of the sparse vector.

 {\it Step~3: } For the last step, we find a good upper bound on
 $\kappa_{\infty}(M_i)$. We know that $\kappa_{\infty}(M_i) \leq
 \sqrt{k}\cdot\kappa(M_i), $ where $\kappa(\cdot)$ denotes the
 condition number with respect to the $\ell_2$-norm. Thus we
 equivalently upper bound $\kappa(M_i)$.  It is straightforward to see that 
 $M_i$ can be decomposed as $M_i = D V_k X_k V_k^\top$, where $X_k$ is
 a $k \times k$ diagonal matrix containing the nonzero coefficients of
 the sparse vector $x$ on its diagonal, $V_k$ is a $k \times k$
 Vandermonde matrix with roots on the unit circle, and $D$ is a
 diagonal (and unitary) $k \times k$ matrix containing appropriate
 powers of the $a_i$.  Obviously, $\kappa(D) = 1$. Moreover, as the
 dynamic range of $x$ is bounded by $2^\ell$, we have that
 $\kappa(X_k) = |x_{\mathrm{max}}|/|x_{\mathrm{min}}| \leq 2^\ell.$
 Moreover, we use the following lemmas:

\begin{lem}\cite{Vandermonde} \label{vandercond}
  The condition number of any complex $k \times k$ Vandermonde matrix $V$
  with roots on the unit circle is at most $\sqrt{2k}$.
\end{lem}

\begin{lem}
  If $Q$ and $R$ are square matrices with complex entries, then $\kappa(QR) \leq
  \kappa(Q)\kappa(R)$.
\end{lem}

The second lemma is easy to derive and we omit its proof. Altogether, 
we conclude that $\kappa(M_i) \leq k 2^{\ell+1}$, which, combined
with Step~2, implies that it suffices to have the $y_i$
within $O\left(k\log{\frac{n}{k}}+ \ell\right)$ bits to correctly reconstruct the
support of $x$.

After we find the correct support, the reconstruction problem is
reduced to a $k \times k$ system of linear equations defined by $k$
columns of the Vandermonde measurement matrix and the corresponding
measurement outcomes. As $k$ columns of a Vandemonde matrix also form a
Vandemonde matrix, by Lemma~\ref{vandercond} the condition number of
the matrix defining the equations is at most
$\sqrt{2k}$. Hence, again using Corollary~\ref{num}, knowing the measurement outcomes within $O(\log
k+\ell)$ bits would be sufficient for obtaining $\ell$ bits of
precision in the reconstruction of $x$. However, this number is less
than the bound that we derived before for finding the correct support of $x$. 
This concludes the proof.

 \section{A General Bound on Precision}\label{comparison}

 In the preceding section we obtained a bound on the amount of
 precision required for the measurements obtained from a Vandermonde
 measurement matrix for ensuring reliable recovery of the sparse
 signal using a particular reconstruction algorithm, namely, syndrome
 decoding.  In this section, we consider a similar problem, but for a
 general class of measurement matrices satisfying a suitable
 \emph{restricted isometry property (RIP)} and considering convex
 optimization as the recovery method.  
%% extra comment:
We remark that, as in the case
 of Vandermonde measurements, our main focus here is on the amount of
 \emph{information} that needs to be extracted from the measurement outcomes,
 and we do not take the imprecision of numerical computations into account.
 In particular, we assume that the reconstruction algorithm uses an
 idealized computation model, but receives truncated measurement
 outcomes on its input.
%%%%%%%%%%%%%%%%%

 The bound that we obtain in this section is a direct corollary of a
 result by Candes et al.\ on robust recovery of sparse signals from
 inaccurate measurements \cite{Candes-Stable}. We begin by recalling
 this result and the required notation.  For a complex $m \times n$
 matrix $A$ (where $m \leq n$) and positive integer $s$, denote by the
 \emph{$s$-isometry constant $\delta_s$} the infimum over all choices
 of $\delta$ that satisfy $ (1-\delta) \| c \|_2^2 \leq \| A' c \|_2^2
 \leq (1+\delta) \| c \|_2^2, $ for every $m \times s$ submatrix $A'$
 of $A$ and every $c \in \C^s$.  The following is the main result
 proved in \cite{Candes-Stable}:

 \begin{thm} \label{l1stab} Suppose that the measurement matrix $A$
   satisfies the restricted isometry property that $\delta_{3s} + 3
   \delta_{4s} < 2$, for some positive integer $s$. Then for every
   $\eps > 0$, every $s$-sparse signal $x \in \C^n$, and $y := Ax +
   e$, $\| e \|_2 \leq \eps$, the solution $\hat{x}$ to the convex
   program
   \[
   (P_2) \qquad \min \|x'\|_1, \qquad \qquad \|A x'-y\|_2 \leq \epsilon,
   \]
   satisfies $\| \hat{x} - x \| \leq C \eps$, where $C$ is a
   positive constant only depending on $\delta_{4s}$. %\qed
 \end{thm}
 % The theorem states that one can recover any sparse signal within an
 % error whose size is proportional to the noise level in the
 % measurements. It is also pointed out in \cite{Candes-Stable} that
 % this is the best one can hope for. Additionally it is

 As pointed out in \cite{Candes-Stable}, the program $(P_2)$ can be
 solved efficiently using known techniques from convex optimization.
 % Hence, our upper bound on $\ell'$ is valid with respect to a solver
 % for $(P_2)$ without considering a particular implementation of
 % it. Nevertheless, the bound shows the amount of precision that is
 % information-theoretically sufficient for a satisfactory
 % reconstruction of the signal.
% 
 Now suppose in the sequel that $A$ is an $m \times n$ matrix
 satisfying the isometry property needed above, and that we wish to
 reconstruct an approximation $\hat{x}$ of a $k$-sparse signal $x \in
 \C^n$ within at least $\ell$ significant bits (in the fixed-point
 model) from $y := Ax + e$, where $e$ is the rounding error, such that
 $\hat{x}$ has the same support as that of $x$. Obviously, for that to
 be possible the dynamic range of $x$ must be at most $2^\ell$, as
 otherwise a nonzero but small coefficient of $x$ might be confused
 with zero.
 % Thus, we would like to satisfy
 % \begin{equation}
 %   \frac{\|\tilde{x}-x\|_{\infty}}{\|x\|_{\infty}} <
 %   2^{-\ell}.\label{relx}
 % \end{equation}
 The following straightforward corollary of Theorem~\ref{l1stab}
 quantifies the amount of precision needed for $y$:

 \begin{cor} \label{coro} Let $e \in \C^m$ denote the quantization
   error in $y$. In order to ensure that $\|\hat{x}-x\|_{\infty}/
   \|x\|_{\infty} < 2^{-\ell}$, it suffices to have
   \begin{equation*}
     \frac{\|e\|_{\infty}}{\|y\|_{\infty}} < \frac{2^{-\ell}}{C k \sqrt{m}}. \label{err1}
   \end{equation*}
 \end{cor}

 The result states that if we have the measurements within $O(\ell
 +\log k + \log m)=O(\ell + \log m)$ bits of precision, we can ensure
 that the program $(P_2)$ obtains a reconstruction that approximates
 $x$ within $\ell$ bits of precision, which in particular implies
 correct support recovery of $x$. Hence, the total number of bits
 needed from the measurements can be upper bounded by $O(m(\ell + \log
 m))$.

 As a concrete example, consider a measurement matrix $A$ that outputs
 a set of $m$ random Fourier coefficients of the signal. It is shown
 in \cite{Candes-Stable} that, in order for $A$ to satisfy the RIP
 needed by Corollary~\ref{coro} with overwhelming probability, it is
 sufficient to take $m = O(k (\log n)^6)$. Thus in this case, $O(k
 (\log n)^6 (\ell + \log k + \log\log n))$ bits from the measurement
 vector $y$ would be sufficient to reconstruct $x$ within precision
 $\ell$. On the other hand, the upper bound that we obtained for the
 Vandermonde matrix with syndrome decoding is a total of $O(k\ell +
 k^2 \log (n/k))$ bits from $y$. The two bounds are incomparable, but
 the latter is better for very sparse signals (e.g., $k = O(\log n)$).

 \bibliographystyle{IEEEtran} \bibliography{bibliography}

\end{document}